\RequirePackage{cmap}
\documentclass[12pt, a4paper]{article}
\usepackage[cp1251]{inputenc}  
\usepackage[T2A]{fontenc} 
\usepackage[a4paper,hmargin=2.5cm,vmargin=2.5cm]{geometry}     
\usepackage[english,russian]{babel}
\usepackage{latexsym} % \lhd \rhd
\usepackage{amsfonts} % \mathbb

\usepackage{xcolor}
\usepackage{hyperref}[]
\definecolor{linkcolor}{HTML}{000000} % цвет ссылок
\definecolor{urlcolor}{HTML}{0000EE} % цвет гиперссылок
\definecolor{citecolor}{HTML}{000000} %цвет для \cite
\hypersetup{linkcolor=linkcolor,urlcolor=urlcolor, citecolor=citecolor, colorlinks=true}

\ifx\pdfoutput\undefined
\usepackage{graphicx}
\else
\usepackage[pdftex]{graphicx}
\fi
\usepackage{wrapfig}

\usepackage{amssymb}
\usepackage{amsthm}
\usepackage{amsmath}

\DeclareMathOperator{\Cells}{Cells}

\newtheorem{theorem}{Теорема}
\newtheorem{hypo}{Гипотеза}
\newtheorem{problem}{Задача}

\newtheorem{lemm}{Лемма}
\newtheorem{idea}{Предложение}

\theoremstyle{definition}

\theoremstyle{remark}
\newtheorem{bytheway}{Замечание}

\newcommand{\eqdef}{\stackrel{\mathrm{def}}{=}}
\newcommand{\toto}{\rightrightarrows}

\newcommand*{\hm}[1]{#1\nobreak\discretionary{}%
	{\hbox{$\mathsurround=0pt #1$}}{}}
\newcommand*{\hide}[1]{}

\begin{document}

%Сведения об авторе

%Федоров Михаил Сергеевич; Национальный исследовательский институт Высшая школа экономики, Факультет математики, 119048, Москва, ул. Усачева, 6; +79166082216; fedorov.mikhail.s@gmail.com: M. Fedorov
%	\thispagestyle{empty}
%	\pagebreak
	
%	\author{Михаил Федоров\\ Национальный исследовательский институт Высшая школа экономики,\\ Факультет математики, fedorov.mikhail.s@gmail.com}
%	\title{Некоторые особенности протекания нескольких жидкостей на шестиугольной решетке.\\
%	%Conservation lows in quantum field theory on graphs.
%}
%	\date{}
%	\maketitle
\author{Федоров Михаил Сергеевич\footnote{fedorov.mikhail.s@gmail.com, Национальный исследовательский институт Высшая школа экономики, Факультет математики}}
\title{Некоторые особенности распределения вероятностей протекания нескольких жидкостей на шестиугольной решетке.\\
%Some aspects of probability distribution for percolation of several fluids on the hexagonal lattice.
}

\date{}
\maketitle
\emph{Аннотация.}
%В работе рассмотрено обобщение перколяции на шестиугольной решетке на несколько жидкостей. Доказан аналог центральной предельной теоремы для величин-индикаторов протекания каждой из жидкостей, которые не являются независимыми в совокупности. Приводятся гипотезы о других свойствах протекания нескольких жидкостей на основе численных экспериментов.
В работе рассматривается равномерная случайная раскраска клеток шестиугольной решетки в \(2^{n-1}\) цветов (стандартная модель Поттса при бесконечной температуре), которую можно рассматривать как обобщение перколяции на \(n\) жидкостей --- попарно-независимых, но зависимых в совокупности. В этой модели вводится новая наблюдаемая, которую можно интерпретировать как долю протекающих жидкостей. Для этой наблюдаемой доказывается аналог центральной предельной теоремы и формулируется несколько гипотез на основе численных экспериментов.

\emph{Ключевые слова и фразы:} перколяция, центральная предельная теорема, шестиугольная решетка, гауссово распределение, модель Поттса.

\section{Введение}
\begin{wrapfigure}[19]{r}{250pt}
	\includegraphics[scale=0.65]{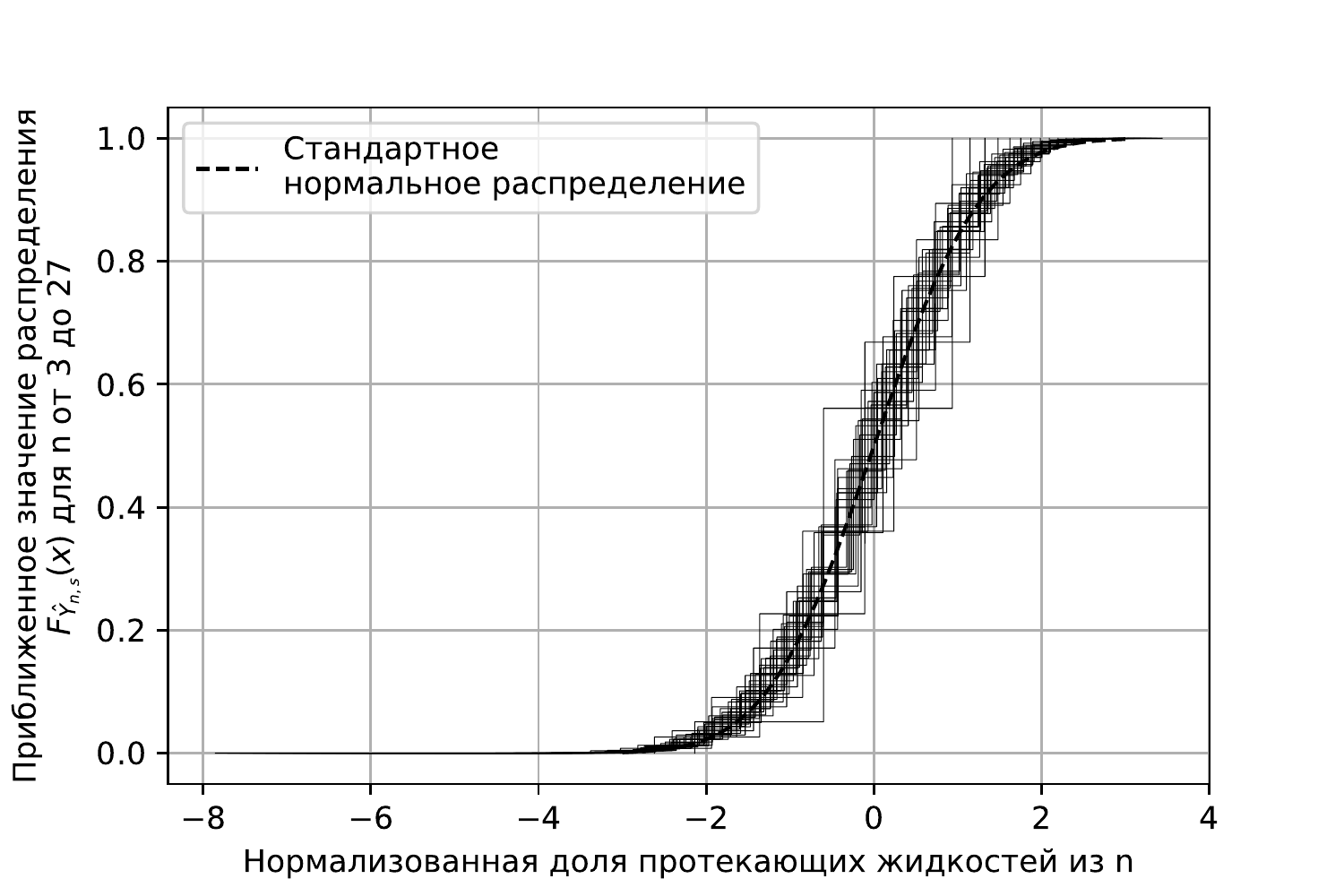}
	\caption{К теореме \ref{thm::LikeCLT}. Распределения величин \(\frac{\sqrt{n}(\overline X_{n, s} - p_s)}{\sqrt{p_s(1-p_s)}}\) сходятся к нормальному.}\label{fig::MainRes}
\end{wrapfigure}
 
Теория просачивания появилась во второй половине 20--го века. Она используется для моделирования различных физических процессов. Например, задачи перколяции возникают при исследовании свойств лесных пожаров и ферромагнетиков \cite{Intoduct,Efros}. Для обычной перколяции --- одной жидкости --- получено множество результатов, например, теорема Кестена \cite{Kesten} о том, что вероятность протекания между центром и границей круга стремится к нулю при увеличении радиуса. Обширный обзор и краткую историю теории перколяции можно найти в \cite{60_years}.

%\begin{wrapfigure}[19]{r}{250pt}
%	\includegraphics[scale=0.65]{3_liquids.pdf}
%	\caption{Изменение вероятности протекания жидкости при росте размера решетки. К гипотезе \ref{hp::grid}.}\label{fig::s_to_infty}
%\end{wrapfigure}

В работе исследуется модель протекания нескольких жидкостей, которая при минимальном их количестве совпадает с обычной перколяцией. Эта модель эквивалентна хорошо изученной стандартной модели Поттса \cite[\S 2]{60_years} с \(2^{n-1}\) состояниями при бесконечной температуре (знакомство с моделью Поттса для понимания статьи не требуется), в частности вероятностное пространство этой модели --- множество раскрасок клеток решетки в \(2^{n-1}\) цветов.

\begin{wrapfigure}[20]{r}{250pt}
	\includegraphics[scale=0.65]{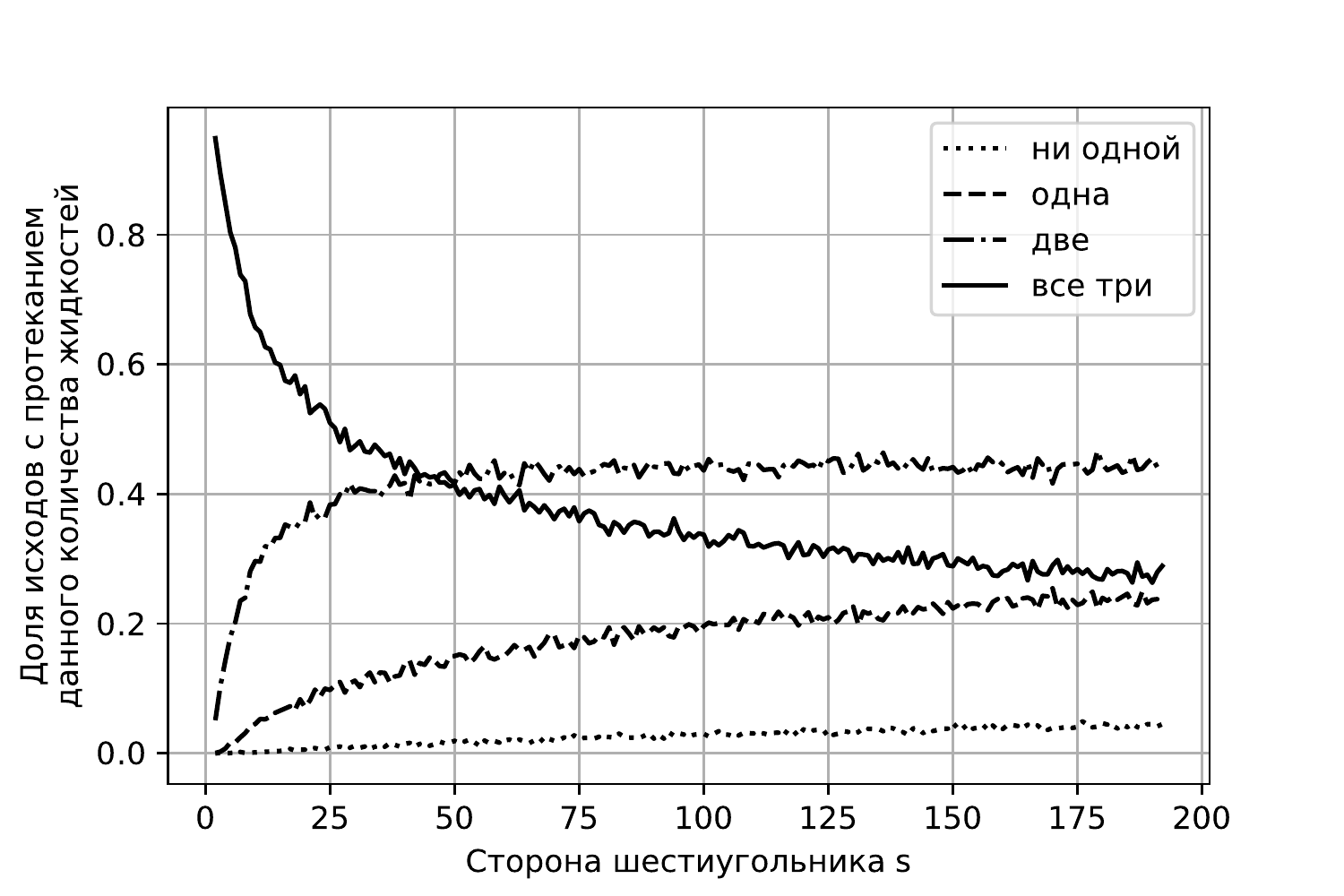}
	\caption{К гипотезе \ref{hp::grid}. Изменение вероятности протекания данного количества жидкостей при росте размера шестиугольника, если всего три жидкости.}\label{fig::s_to_infty}
\end{wrapfigure}

В данной работе вводится простейшая наблюдаемая, не ``сводящаяся'' к протеканию одной жидкости, --- доля протекающих жидкостей. Для распределения этой величины доказан аналог центральной предельной теоремы (теорема \ref{thm::LikeCLT} ниже).
Численные эксперименты (см. Рис. \ref{fig::MainRes} и \S\ref{sec::exper}) подтверждают этот результат.

Этот результат интересен, потому что рассматриваемые наблюдаемые попарно независимы, но зависимы в совокупности. Вообще говоря, для зависимых в совокупности величин центральная предельная теорема не выполняется,  красивые контрпримеры приводятся в \cite{CLT_fails_pairwise, CLT_fails_5-tuplewise}. Также этот результат перекликается с решением \(SU(N)\) калибровочной теории при \(N\to\infty\) \cite{Poliakov}.

В \S\ref{sec::res} вводятся определения и обозначения, необходимые для формулировки теоремы \ref{thm::LikeCLT}. В \S\ref{sec::my_prf} приводится доказательство это теоремы. В \S\ref{sec::exper} описан численный эксперимент и сформулированы задача \ref{hyp::IndependInLimit} и гипотеза \ref{hp::grid} (см. Рис. \ref{fig::s_to_infty}). В \S\ref{sec::big_sum} приводится формализм, который, надеемся, может пригодится для решения этих задач.

\section{Основной результат} \label{sec::res}

\begin{wrapfigure}[12]{r}{120pt}
	\includegraphics[scale=0.25]{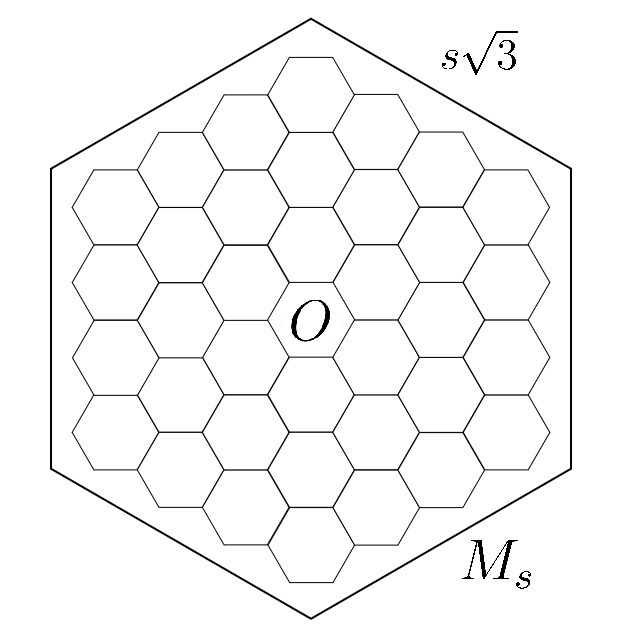}
	\caption{Множество \(M_s\)}\label{fig::grid}	
\end{wrapfigure}

	Для формулировки основного результата нам понадобятся следующие определения и обозначения (Рис. \ref{fig::grid}). 
	Для каждого целого \(s\geqslant2\) рассмотрим множество \(M_s\) всех клеток шестиугольной решетки со стороной клетки 1, расположенных внутри правильного шестиугольника со сторонами длины \(s\sqrt{3}\), перпендикулярными сторонам клеток, и с центром в центре клетки \(O\). Множество \(M_s\) содержит \(m\hm+1\hm\eqdef1+3s(s-1)\) клеток. Клетки, кроме центральной \(O\), пронумеруем: \(v_1, \ldots, v_m\). \textit{Соседними} клетками будем называть те, у которых есть общая сторона. \textit{Граничными} назовем те клетки, у которых в множестве \(M_s\) менее шести соседних.

	Для каждого целого \(n\geqslant2\) обозначим \textit{множество раскрасок клеток (кроме центральной) в \(2^{n-1}\) цветов}
\begin{equation*}
\Omega_{n, s}\eqdef\left\lbrace(f_1, \ldots, f_n) \colon \{v_1, \ldots, v_m\}\to \{0, 1\}^n\mid f_1 + \cdots + f_n \equiv 1\pmod2\right\rbrace.
\end{equation*}
	Заметим, что \(\left|\Omega_{n, s}\right| = 2^{(n-1)m}\).
	Элементы множества \(\Omega_{n, s}\) будем называть \textit{раскрасками} в \(2^{n-1}\) цветов. Смысл такого определения сейчас станет ясен.

	В дальнейшем мы будем рассматривать \( \Omega_{n, s}\) как вероятностное пространство с мерой \(P(A) = |A|2^{-(n-1)m}\) для всех \(A\subset\Omega_{n, s}\) 
	(оно изоморфно \textit{стандартной модели Поттса} с \(2^{n-1}\) состояниями при бесконечной температуре).
	
	\textit{Путем \(S\) от центра до границы} назовём последовательность различных клеток \(v_{j_1}, \ldots, v_{j_t}\) такую, что клетка \(v_{j_1}\) --- соседняя с \(O\), клетки \(v_{j_l}\) и \(v_{j_{l+1}}\) --- соседние для каждого \(l\geqslant1\), и клетка \(v_{j_t}\) является граничной.
	Пусть \(i\in\{1, \ldots, n\}\). Будем говорить, что \textit{\(i\)--я жидкость протекает от центра до границы} для данной раскраски \(f = (f_1, \ldots, f_n)\in\Omega_{n, s}\), если существует путь \(S=(v_{j_1}, \ldots, v_{j_t})\) от центра до границы, такой что для всех \(l=1,\ldots,t\) выполнено \(f_i(v_{j_l})=1\).
	
	\textit{Протеканием \(i\)--ой жидкости} назовем событие 
	\begin{equation*}
	A_{n, s, i}\eqdef\{f\in\Omega_{n, s}\mid i\text{--я жидкость протекает от центра до границы для раскраски }f\}.
	\end{equation*}
	Несложно видеть, что вероятность протекания жидкости не зависит ни от ее номера \(i\), ни от общего количества жидкостей \(n\). А именно, для всех \(s, n, q, i, j\) верно \(P(A_{n, s, i})=P(A_{n, s, j})\)  и \(P(A_{n, s, i})=P(A_{q, s, i})\). Обозначим эту вероятность \(p_s\eqdef P(A_{2,s,1})\). Кроме того, при \(n\geqslant3\) события протекания жидкостей с различными номерами попарно независимы (например, по лемме \ref{lem::UnionIndependance} ниже), но, вообще говоря, зависимы в совокупности (например, для \(n=3\) и \(s=2\)). Именно для того, чтобы присутствовала нетривиальная зависимость жидкостей в совокупности, мы рассматриваем раскраски в \(2^{n-1}\), а не в \(2^n\) цветов.

Определим случайную величину \(X_{n ,s, i}\colon\Omega_{n, s}\to\{0, 1\}\) формулой

	\begin{equation*}
	X_{n ,s, i}(f) \eqdef \begin{cases}
	1,& \substack{\text{если \(i\)--ая жидкость протекает от центра }\\ \text{до границы для раскраски \(f\);}}\\
	0,& \text{иначе}.
	\end{cases}
	\end{equation*} 
	
	Назовем \(\overline X_{n,s} \eqdef \frac{1}{n}\sum_{i=1}^n X_{n, s, i}\) \textit{долей протекающих до границы жидкостей}.
	
\begin{theorem}[Центральная предельная теорема для доли протекающих жидкостей] \label{thm::LikeCLT}
	Для любого \(s\) функция распределения случайной величины \(\dfrac{\sqrt{n}(\overline X_{n, s} - p_s)}{\sqrt{p_s(1-p_s)}}\) сходится равномерно к стандартному нормальному распределению при \(n\to\infty\).
	%	\begin{equation}
	%	\frac{\sqrt{n}\left(\overline X_{n,s} - \mathrm EX_{2, s, 1}\right)}{\sqrt{\mathrm DX_{2,s,1}}}\stackrel{d}{\to} N(0,1).
	%	\end{equation}
\end{theorem}

Численные эксперименты (см. \S\ref{sec::exper}) подтверждают этот результат (Рис. \ref{fig::MainRes}).

\section{Доказательство}\label{sec::my_prf}
Для доказательства теоремы \ref{thm::LikeCLT} невозможно сразу применить центральную предельную теорему, потому что у нас нет одной последовательности случайных величин, а есть последовательность их наборов. Поэтому воспользуемся независимостью в совокупности всех величин, кроме одной, и покажем, что оставшаяся <<слабо>> влияет на функцию распределения. При этом мы не будем оценивать <<степень независимости>> всего набора величин --- сделать это сложнее, чем доказать саму теорему \ref{thm::LikeCLT} \cite{Novikov}.

Для доказательства теоремы \ref{thm::LikeCLT} понадобятся три несложные леммы.
\begin{lemm} \label{lem::UnionIndependance}
	Для любых \(n, s\) любое подмножество мощности \(n-1\) набора событий \(\{A_{n, s, 1}, \ldots, A_{n, s, n}\}\) есть набор событий, независимый в совокупности.
\end{lemm}
\begin{proof}
	Так как жидкости можно перенумеровать, то без ограничения общности рассмотрим набор \(\{A_{n, s, 1},\ldots,A_{n, s, n-1}\}\) и проверим, что для любого \(k=2,\ldots,n-1\) выполнено
	\begin{align*}
	P(A_{n, s, 1}\cap\ldots\cap A_{n, s, k})=P(A_{n, s, 1})\cdots P(A_{n, s, k}).
	\end{align*}
%\begin{lemm} \label{lem::UnionIndependance}
%	Для любых \(n, s\) любое собственное подмножество набора событий \(\{A_{n, s, 1}, \ldots, A_{n, s, n}\}\) есть набор событий, независимый в совокупности.
%\end{lemm}
%\begin{proof}
%	Достаточно доказать лемму для подмножеств мощности \(n-1\), так как любое другое подмножество можно расширить до большего по мощности. Так как величины можно перенумеровать, то достаточно рассмотреть подмножество \(\{X_{n, s, 1},\ldots,X_{n, s, n-1}\}\) и для каждого \(k\in\{0, 1,\ldots,n\}\) проверить, что 
%	\begin{align*}
%	P(\{X_{n, s, 1}=1\}\cup\ldots\cup\{X_{n, s, k}=1\}\cup\{X_{n, s, k+1}=0\}\cup\ldots\cup\{X_{n, s, n}=0\})=\\
%	=P(\{X_{n, s, 1}=1\})\cdots P(\{X_{n, s, k}=1\})P(\{X_{n, s, k+1}=0\})\cdots P(\{X_{n, s, n}=0\}).
%	\end{align*}
Протекание \(i\)--ой жидкости при раскраске \(f\) определяется функцией \(f_i\).
Легко видеть, что для каждого \(i\) имеется \(2^mp_s\) функций \(f_i\colon\{v_1,\ldots,v_m\}\to\{0,1\}\) таких, что \(i\)--ая жидкость протекает из центра до границы при раскраске \((f_1, \ldots, f_{i-1}, f_i, f_{i+1}, \ldots, f_n)\in\Omega_{n, s}\) для некоторого выбора функций \(f_1, \ldots, f_{i-1}, f_{i+1}, \ldots, f_n\). Тогда имеется \((2^m)^{n-1}(p_s)^{k}\) наборов функций \((f_1,\ldots,f_{k}, f_{k+1}, \ldots, f_n)\) таких, что жидкости \(1,\ldots,k\) протекают, так как для каждого набора \((f_1,\ldots,f_{k})\) функции \(f_{k+1},\ldots,f_{n-1}\) можно выбирать произвольно, а функция \(f_n\) восстанавливается однозначно по формуле
	\begin{equation*}
	f_n(v)=
	f_1(v)+\cdots+f_{n-1}(v)\pmod2.
	\end{equation*}
	Значит, \(P(A_{n, s, 1}\cap\ldots\cap A_{n, s, k}) \hm= 2^{m(n-1)}p_s^{k}/2^{m(n-1)} \hm=p_s^{k}\).
	Так как для каждого \(i\) верно \(P(A_{n, s, i}) = p_s\), то получаем требуемое равенство.
\end{proof}

\textit{Одновременным протеканием ровно \(k\) жидкостей из \(n\)} назовем событие 
\begin{equation*}
B_{n, s, k}\eqdef\left\lbrace f\in\Omega_{n, s}\left|\, \sum_{i=1}^n X_{n, s, i}(f) = k\right.\right\rbrace.
\end{equation*}

Легко видеть, что \(p_s = P(B_{2, s, 2})+\frac{1}{2}P(B_{2, s, 1})\).

\begin{lemm} \label{lem::base}		
	Для всех целых \(s\geqslant2\), \(n\geqslant3\) и \(1\leqslant i\leqslant n\) имеем 
	\begin{align*}
	\mathrm E \overline X_{n, s} = \mathrm EX_{n, s, i} &= p_s,\\
	\mathrm D \overline X_{n, s} = \mathrm DX_{n, s, i} &= p_s(1-p_s),\\
	\mathrm E|X_{n, s, i}-p_s|^3&=p_s(1-p_s)(1-2p_s+2p_s^2)
	<\infty.
	\end{align*}

	Для всех \(s\) имеем \(\lim\limits_{n\to\infty}\max\limits_k P(B_{n, s, k}) = 0\).
	
	Для всех \(k, n\) имеем \(\lim\limits_{s\to\infty}P(B_{n, s, k}) = 0\).
\end{lemm}
\begin{bytheway}
	Последнее утверждение не используется для доказательства теоремы \ref{thm::LikeCLT}. В численном эксперименте при \(n=3, s<200\) стремления к нулю вероятностей \(P(B_{n, s, k})\) не видно. А именно, на рисунке \ref{fig::s_to_infty} эти вероятности возрастают для каждого \(k<3\).
\end{bytheway}
\begin{proof}

Во-первых, из линейности матожидания \(\mathrm E \overline X_{n, s} \hm= \mathrm EX_{n, s, i}=p_s\). А так как по лемме \ref{lem::UnionIndependance} для всех \(n\geqslant3, s\geqslant2\) события \(A_{n, s, i}\) попарно независимы, то \(\mathrm D \overline X_{n, s} \hm= \mathrm DX_{n, s, i}\hm= p_s(1-p_s)\). Также
%В силу той же леммы, при для всех \(n\geqslant4, s\geqslant2\) любые три события \(A_{n, s, i}\) независимы в совокупности, а значит третий момент среднего арифметического этих случайных величин равен среднему арифметическому их третьих моментов, откуда
\begin{equation*}
\mathrm E|X_{n, s, i}-p_s|^3\hm=p_s(1-p_s)^3+(1-p_s)p_s^3\hm=p_s(1-p_s)(1-2p_s+2p_s^2)<\infty.
\end{equation*}
%\begin{equation*}
%\mathrm E|\overline X_{n, s}-p_s|^3\hm<(n+1)*\max\{p_s^3, (1-p_s)^3\}<\infty.
%\end{equation*}

Во-вторых, количество элементов множества \(\Omega_{n, s}\), для которых от центра до границы протекают ровно \(k\) жидкостей из \(n\), равно \(2^{m(n-1)} P(B_{n, s, k})\). 
Рассмотрим естественное вложение множества \(\Omega_{n, s}\) в \(\Omega_{n+1, s}\): элемент \((f_1, \ldots, f_n)\in\Omega_{n, s}\) переходит в элемент \((f_1, \ldots, f_{n+1})\in\Omega_{n+1, s}\), где \(f_{n+1}\) --- тождественно нулевая функция.  (Неформально, мы забываем про ограничение \(f_1+\cdots + f_n\equiv1\pmod2\).) В силу леммы \ref{lem::UnionIndependance}, ровно \(k\) жидкостей из первых \(n\) протекает для \(2^{mn}{n\choose k}p_s^k(1\hm-p_s)^{n-k}\) раскрасок в \(2^{n}\) цветов. В силу того, что при вложении раскраски в \(2^{n-1}\) цветов, для которых протекает ровно \(k\) жидкостей из \(n\), отображаются в раскраски в \(2^{n}\) цветов, для которых протекает ровно \(k\) жидкостей из первых \(n\), то \(2^{m(n-1)} P(B_{n, s, k})\leqslant2^{mn}{n\choose k}p_s^k(1\hm-p_s)^{n-k}\). А тогда
%Есть естественное вложение \(\Omega_{n, s}\) в \(\oplus_{i=1}^{n}\Omega_{2, s}\): для каждого элемента \((f_1, \ldots, f_n)\in\Omega_{n, s}\) каждая компонента \(f_l\) задает клетки, через которые протекает первая жидкость в \(l\)--ом прямом слагаемом. Но для \(\oplus_{i=1}^{n}\Omega_{2, s}\) первая жидкость протекает ровно в \(k\) прямых слагаемых в \(2^{mn}{n\choose k}p_s^k(1\hm-p_s)^{n-k}\) случаях. В силу того, что при вложении элементы множества \(\Omega_{n, s}\), для которых протекают ровно \(k\) жидкостей из \(n\), отображаются в элементы \(\oplus_{i=1}^{n}\Omega_{2, s}\), для которых первая жидкость протекает ровно в \(k\) прямых слагаемых, то \(2^{m(n-1)} P(B_{n, s, k})\leqslant2^{mn}{n\choose k}p_s^k(1\hm-p_s)^{n-k}\). А тогда
\begin{equation*}
\max_kP(B_{n, s, k})\hm\leqslant2^m\max_k\left({n\choose k}p_s^k(1\hm-p_s)^{n-k}\right)\to0 \text{ при }n\to\infty,
\end{equation*}
где последнее следует из теоремы Муавра--Лапласа \cite[том 1, \S. VII.3]{Feller:B-E_inq}.
	
В-третьих, по определению одновременного протекания,
\begin{align*}
P(B_{n, s, k}) &= \sum_{1\leqslant i_1<\cdots<i_k\leqslant n}P\left(A_{n, s, i_1}\cap\ldots\cap A_{n, s, i_k}\setminus\bigcup_{j\notin\{i_1,\ldots,i_k\}}A_{n, s, j}\right)\leqslant\\
&\leqslant\sum_{1\leqslant i_1<\cdots<i_k\leqslant n}P\left(A_{n, s, i_1}\cap\ldots\cap A_{n, s, i_k}\right) = {n\choose k}P(A_{n,s,1}\cap\ldots\cap A_{n, s, k})\leqslant\\
&\leqslant {n\choose k}P(A_{n, s, k}) \hm= {n\choose k}p_s\to0 \text{ при }s\to\infty,
\end{align*}
%\bigcup_{j\notin\{i_1,\ldots,i_k\}}A_{n, s, j}
%Так как $$P\left(A_{n, s, i_1}\cap\ldots\cap A_{n, s, i_k}\setminus\bigcup_{j\notin\{i_1,\ldots,i_k\}}A_{n, s, j}\right)\hm\leqslant P(A_{n, s, i_1}\hm\cap\ldots\cap A_{n, s, i_k})\hm=P(A_{n, s, 1}\cap\cdots\cap A_{n, s, k}),$$ то
%\begin{equation*}
%P(B_{n, s, k}) \hm\leqslant {n\choose k}P(A_{n,s,1}\cap\ldots\cap A_{n, s, k}) \hm\leqslant {n\choose k}P(A_{n, s, k}) \hm= {n\choose k}p_s\to0 \text{ при }s\to\infty,
%\end{equation*}
где последнее доказывается аналогично теореме Кестена~\cite[теор. 2]{Kesten_edges, Kesten}.
\end{proof}

Для фиксированного \(s\) определим случайные величины \(Y_{n, k} = X_{n, s, k} \hm- p_s\) для всех \(k=1,\ldots, n\). Обозначим \(\overline Y_{n, s} \eqdef \overline X_{n, s} - p_s = \frac{1}{n}\sum_{i=1}^n Y_{n, i}\) и \(\overline Z_{n, s} = \frac{1}{n-1}(Y_{n, 1}+\cdots+Y_{n, n-1})\).

Будем обозначать через \(F_X\) функцию распределения случайной величины \(X\), а через \(N(x)\) --- стандартное нормальное распределение.

\begin{lemm} \label{lem::UnifLim}
	Для каждого фиксированного \(s\) верно \(F_{\overline Z_{n,s}} - F_{\overline Y_{n,s}}\toto0\) на \(\mathbb R\) при \(n\to\infty\).
\end{lemm}
\textit{Доказательство.}(См. Рис. \ref{fig::for_proof})
	Обозначим через \(C_{n,s,k}\hm\eqdef\Omega_{n, s}\setminus A_{n,s,k}\) --- событие ``непротекания'' \(k\)--ой жидкости.
	Заметим, что так как \(\overline Z_{n, s}\) --- дискретная величина, принимающая значения \(-p_s, \frac{1}{n-1}\hm-p_s,\ldots, \frac{n-2}{n-1}\hm-p_s, 1\hm-p_s\), то \(F_{\overline Z_{n,s}}\) постоянна на всех полуинтервалах \(\left[\frac{k-1}{n-1}-p_s,\frac{k}{n-1}-p_s\right)\), где \(k\in\{1,\ldots,n-1\}\).
	
	\begin{wrapfigure}[10]{r}{200pt}
		\includegraphics[scale=0.5]{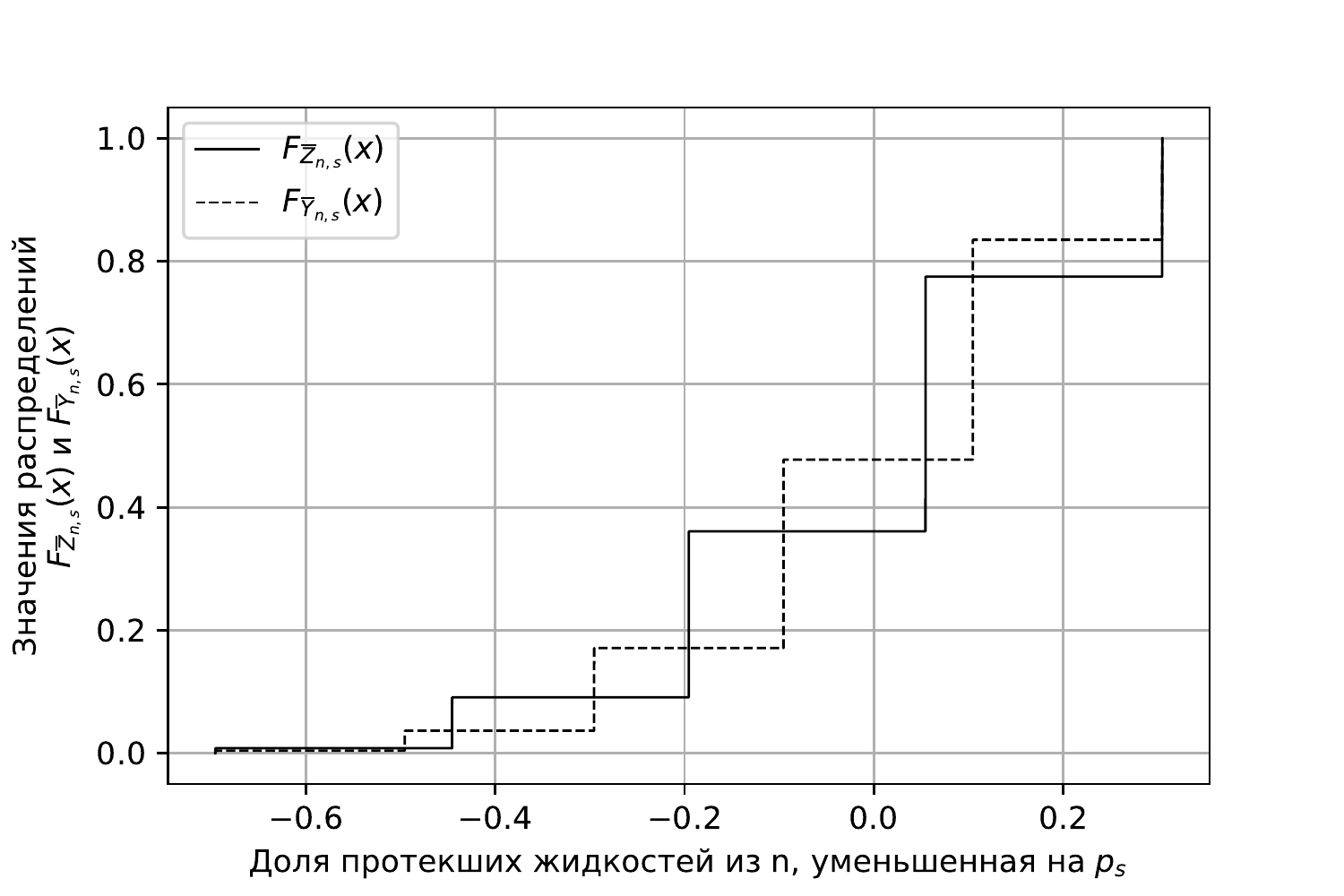}
		\caption{Графики \(F_{\overline Z_{n,s}}\) и \(F_{\overline Y_{n,s}}\) }\label{fig::for_proof}
	\end{wrapfigure}
	Тогда, так как \(\frac{k}{n}\hm\in\left[\frac{k-1}{n-1}, \frac{k}{n-1}\right)\), то
	\begin{multline*}
	F_{\overline Z_{n,s}}\left(\frac{k}{n}-p_s\right)\hm-F_{\overline Y_{n,s}}\left(\frac{k}{n}-p_s\right)\hm=\\
	=F_{\overline Z_{n,s}}\left(\frac{k-1}{n-1}-p_s\right)\hm-F_{\overline Y_{n,s}}\left(\frac{k}{n}-p_s\right).
	\end{multline*}
	Кроме того, по определению, \(F_{\overline Z_{n,s}}\left(\frac{k-1}{n-1}-p_s\right)\) равно вероятности того, что среди первых \(n\hm-1\) жидкостей протекает не больше \(k-1\), а именно
	\begin{multline*}
	F_{\overline Z_{n,s}}\left(\frac{k-1}{n-1}-p_s\right)=\sum_{0\leqslant j\leqslant k}(P(B_{n, s, j}\cap A_{n, s, n})+\\+P(B_{n, s, j-1}\cap C_{n, s, n})),
	\end{multline*}
	где считаем \(B_{n, s, -1} = \varnothing\), и аналогично,
	\begin{multline*}
	F_{\overline Y_{n,s}}\left(\frac{k}{n}-p_s\right)= \sum_{0\leqslant j\leqslant k}P(B_{n, s, j}) 	= \sum_{0\leqslant j\leqslant k}(P(B_{n, s, j}\cap A_{n, s, n})+P(B_{n, s, j}\cap C_{n, s, n})).
	\end{multline*}
Вычитая, находим
	\begin{equation*}
	F_{\overline Z_{n,s}}\left(\frac{k-1}{n-1}-p_s\right)\hm-F_{\overline Y_{n,s}}\left(\frac{k}{n}-p_s\right) = -P(B_{n, s, j}\cap C_{n, s, n})
	\end{equation*}
	Таким образом, получаем
	\begin{equation*}
	0\geqslant F_{\overline Z_{n,s}}\left(\frac{k}{n}-p_s\right)\hm-F_{\overline Y_{n,s}}\left(\frac{k}{n}-p_s\right)%\hm=\\=-P(B_{n, s, k}\cap\{X_{n, s, n} = 1\})\hm
	\geqslant-P(B_{n, s, k})\geqslant-\max_kP(B_{n, s, k}).
	\end{equation*}
	Аналогично, так как \(F_{\overline Y_{n,s}}\) постоянна на \(\left[\frac{k}{n}-p_s,\frac{k+1}{n}-p_s\right)\ni \frac{k}{n-1}-p_s\),
	\begin{multline*}
	F_{\overline Z_{n,s}}\left(\frac{k}{n-1}-p_s\right)\hm-F_{\overline Y_{n,s}}\left(\frac{k}{n-1}-p_s\right)\hm=F_{\overline Z_{n,s}}\left(\frac{k}{n-1}-p_s\right)\hm-F_{\overline Y_{n,s}}\left(\frac{k}{n}-p_s\right)\hm=\\
	=P(B_{n, s, k}\cap A_{n, s, n})\hm\leqslant P(B_{n, s, k})\leqslant\max_kP(B_{n, s, k}).
	\end{multline*}
	Так как  обе функции \(F_{\overline Z_{n,s}}\) и \(F_{\overline Y_{n,s}}\) постоянны на полуинтервалах \(\left[\frac{k-1}{n-1}-p_s,\frac{k}{n}-p_s\right)\) и \(\left[\frac{k}{n}-p_s,\frac{k}{n-1}-p_s\right)\) и на лучах \((-\infty, -p_s)\) и \([1-p_s, +\infty)\), то в силу предыдущих выкладок для всех \(x\in\mathbb R\) верно \(|F_{\overline Z_{n,s}}(x) - F_{\overline Y_{n,s}}(x)|\hm\leqslant\max\limits_kP(B_{n, s, k})\). В силу леммы \ref{lem::base} имеем \(\max\limits_kP(B_{n, s, k})\to0\) при \(n\to\infty\), что влечет \(F_{\overline Z_{n,s}} - F_{\overline Y_{n,s}}\toto0\) при \(n\hm\to\infty\) на \(\mathbb R\).
\qed

\begin{proof}[Доказательство теоремы \ref{thm::LikeCLT}]

Воспользуемся леммой \ref{lem::base} и обозначим через \(\sigma^2\eqdef\mathrm DY_{n, i}\hm=\mathrm DX_{n, s, i}\hm=p_s(1-p_s )\) дисперсию случайной величины \(X_{n, s, i}\), а через \(\rho\eqdef\mathrm E|Y_{n,i}|^3%\hm=\mathrm E|X_{n,s,i}-p_s|^3\hm=p_s(1-p_s)(1-2p_s+2p_s^2)
<\infty\).
Введем также случайные величины 
\begin{equation*}
\hat Z_{n, s}\hm=\dfrac{Y_{n, 1}+\cdots+Y_{n, n-1}}{\sigma\sqrt{n-1}}\hm=\dfrac{\sqrt{n-1}}{\sigma}\overline Z_{n, s}\qquad\text{ и }\qquad\hat Y_{n, s}=\frac{\sqrt{n-1}}{\sigma}\overline Y_{n,s}.
\end{equation*}

Так как по лемме \ref{lem::UnionIndependance} любой набор из \(n\hm-1\) случайных величин из набора \(X_{n, s, 1},\ldots,X_{n, s, n}\) независим в совокупности, то то же самое верно для величин из набора \(Y_{n, 1},\ldots, Y_{n, n}\). Кроме того, последние случайные величины одинаково распределены и \(\mathrm E |Y_{n, k}^3|\hm=\rho\hm<\infty\).
Применив неравенство Берри--Эссеена~\cite[том 2, стр. 542, \S 5, теор.1]{Feller:B-E_inq} к набору величин \(Y_{n, 1},\ldots, Y_{n, n-1}\), получим для всех \(x\in\mathbb R\)
\begin{equation*}
\left|F_{\hat Z_{n, s}}(x)-  N(x)\right|\leqslant\frac{3\rho}{\sigma^3\,\sqrt{n-1}}.
\end{equation*}
В частности, % %\(\left|F_n(x)-N(x)\right|\to 0\) при \(n\to \infty\). 
\(F_{\hat Z_{n, s}}(x)\toto   N(x)\) при \(n\to\infty\) на \(\mathbb R\) для каждого фиксированного \(s\).
Но
\begin{equation*}
F_{\hat Z_{n, s}}(x) \hm- F_{\hat Y_{n, s}}(x)\hm=F_{\overline Z_{n,s}}\left(\dfrac{x\sigma}{\sqrt{n-1}}\right) \hm- F_{\overline Y_{n,s}}\left(\dfrac{x\sigma}{\sqrt{n-1}}\right)\toto0,
\end{equation*}
где последнее стремление следует из леммы \ref{lem::UnifLim}.
Так как \(F_{\hat Z_{n, s}}(x)\toto N(x)\), то и \(F_{\hat Y_{n, s}}(x)\toto N(x)\).
Так как \(N(\frac{\sqrt{n}}{\sqrt{n-1}}x)\toto N(x)\), то \(F_{\frac{\sqrt{n}}{\sigma}\overline Y_{n,s}}(x)\toto N(x)\) при \(n\to\infty\).
\end{proof}
\section{Выражение для вероятности протекания жидкостей.}\label{sec::big_sum}
%\(\backslash\backslash\) набросок, не менялся.
Предложим формализм, который, надеемся, может оказаться полезным для решения задачи \ref{hyp::IndependInLimit} ниже.

Зафиксируем число \(s\geqslant2\) и множество \(M_s\).
Рассмотрим \(\mathrm{S} = \{s_1, \ldots, s_k\}\) --- (непустой, неупорядоченный) \textit{набор} попарно различных путей от центра до границы. Так как существует конечное количество путей, то и их наборов тоже конечное количество. Обозначим количество путей в наборе решеткой: \(\#\mathrm{S}\eqdef k\). Обозначим количество клеток шестиугольника, занимаемых набором путей, через \(|\mathrm{S}|\eqdef|s_1\cup\ldots\cup s_k|\).
%Определим \textit{характеристическую функцию} этого множества клеток, для \(x\in M_s\):
%	\begin{equation*}
%\chi_{\mathrm{S}}(x) = 
%\begin{cases}
%	1, \text{ если } x\in s_1\cup\ldots\cup s_k;\\
%	0, \text{ иначе.}
%\end{cases}
%\end{equation*} 
%	
%Определим \(\Col(n, k) \eqdef (1-\delta_{nk})2^{n-k-1}+\delta_{nk}\frac{1-(-1)^n}{2}\) для \(k\leqslant n\), где \(\delta_{nk}\) --- символ Кронекера. Неформально говоря \(\Col(n, k)\) --- это количество цветов в которые может быть ``раскрашена'' клетка, через которую уже иротекают данные \(k\) жидкостей из \(n\).
	
\begin{idea}
	Верно следующее равенство:
	\begin{equation*}
	p_s = -\sum_{ \mathrm{S}\text{ --- набор путей}}(-1)^{\#\mathrm{S}}2^{-|\mathrm{S}|}.
	\end{equation*}
%	Более того, 
%	\begin{equation}
%	P(B_{n, s, n}) =-\frac{1}{2^m} \sum_{\mathrm S_1, \ldots, \mathrm S_n \text{ --- наборы путей}}(-1)^{\#\mathrm S_1+\cdots+\#\mathrm S_n}\prod_{x\in M_s}\Col(n, \sum_{i=1}^n\chi_{\mathrm S_i}(x)).
%	\end{equation}
%	В частности, для \(n = 3\)
%	\begin{equation*}
%	P(B_{3, s, 3}) = -\sum_{\mathrm S_1, \mathrm S_2, \mathrm S_3\text{ --- наборы путей}}(-1)^{\#\mathrm S_1+\#\mathrm S_2+\#\mathrm S_3}2^{-|\text{<клетки, лежащие ровно в одном из Si>}|}4^{-|\text{<-||- более, чем в одном из Si>}|}.
%	\end{equation*}
\end{idea}

\begin{proof}%[Доказательство равенства \eqref{ElemBin}]
	Рассмотрим все пары \((f, \mathrm S)\) из раскраски \(f\) в два цвета и набора \(\mathrm S\) путей такие, что все клетки, по которым проходят пути из набора, пропускают первую жидкость, то есть для каждой клетки \(v\hm\in s_1\cup\ldots\cup s_k\) выполнено \(f_1(v)\hm=1\). Заметим, что для каждой такой пары первая жидкость протекает от центра до границы для раскраски \(f\), так как набор путей \(\mathrm S\) не пуст по определению. Припишем каждой паре \((f, \mathrm S)\) знак \((-1)^{\#\mathrm{S}}\). Посчитаем сумму знаков всех пар двумя способами.
	
	Для данного набора путей \(\mathrm{S}\) есть \(2^{m-|\mathrm{S}|}\) пар \((f, \mathrm{S})\), так как есть \(m-|\mathrm{S}|\) клеток, не принадлежащих ни одному из путей, каждую из которых можно раскрасить в \(2\) цвета. Поэтому сумма знаков всех рассматриваемых пар равна
	\begin{equation}\label{eqt::longcounting}
	\sum_{ \mathrm{S}\text{ --- набор путей}}(-1)^{\#\mathrm{S}}2^{m-|\mathrm{S}|}.
	\end{equation} 
	С другой стороны, для фиксированной раскраски \(f\), если первая жидкость не протекает от центра до границы для \(f\), то нет ни одной пары вида \((f, \mathrm S)\) среди рассматриваемых. Пусть теперь первая жидкость протекает до границы. Пусть есть ровно \(k\) различных несамопересекающихся путей \(s\) таких, что для всех клеток \(v\in s\) первая жидкость протекает через эти клетки, то есть \(f_1(v)=1\). Тогда пар \((f, \mathrm S)\) таких, что \(\#\mathrm S = l\), среди рассматриваемых ровно \({k \choose l}\). Тогда для фиксированного \(f\) сумма знаков всех пар \((f, \mathrm S)\) равна
	\begin{equation} \label{eqt:BinomEq}
	-{k \choose 1} + {k\choose 2} - \cdots + (-1)^{k}{k\choose k} = -1.
	\end{equation}
	Тогда сумма знаков всех рассматриваемых пар равна \(-2^{m}p_s\). Приравняв это число выражению \eqref{eqt::longcounting} и разделив на \(-2^{m}\), получаем требуемое равенство.
%	Рассмотрим сумму \(-\sum_{ \mathrm{S}}(-1)^{\#\mathrm{S}}2^{m-|\mathrm{S}|}\) по всем наборам путей \(\mathrm{S}\). Докажем, что это целое число равно количеству раскрасок в \(2\) цвета, для которых первая жидкость протекает до границы. 
%	Слагаемое \((-1)^{\#\mathrm{S}}2^{m-|\mathrm{S}|}\) для набора путей \(\mathrm S\) сопоставим тем раскраскам, для которых каждая клетка, входящая хотя бы в объединение путей из \(\mathrm S\), пропускает первую жидкость (а не вторую). Таких раскрасок в точности \(2^{m-|\mathrm{S}|}\), так как клетки, не входящие ни в один из путей из \(\mathrm S\), могут пропускать или первую, или вторую жидкость. Тогда каждая раскраска, не содержащая ни одного пути до границы для первой жидкости, не ``посчитана'' ни разу. А каждая раскраска, содержащая \(l\) различных путей до границы, по которым может протечь первая жидкость, ``посчитана''
%	\begin{equation*}
%	l-{l\choose2}+{l\choose3}-\cdots+(-1)^l {l\choose l}=1
%	\end{equation*}
%	раз. В этой сумме слагаемое вида \((-1)^k{l\choose k}\) соответствует тем наборам \(\mathrm S\) путей, состоящих ровно из \(k\) путей, которые все содержатся среди тех \(l\) путей, по которым протекает первая жидкость для рассматриваемой раскраски.
%	Таким образом, каждая раскраска, для которой первая жидкость протекает до границы, посчитана в этой сумме ровно \(1\) раз.
%	А тогда 
%	\begin{equation*}
%	-p_s  =2^{-m}\sum_{ \mathrm{S}\text{ --- набор путей}}(-1)^{\#\mathrm{S}}2^{m-|\mathrm{S}|} = \sum_{ \mathrm{S}\text{ --- набор путей}}(-1)^{\#\mathrm{S}}2^{-|\mathrm{S}|}.
%	\end{equation*}
\end{proof}
%$$p_s^3 = \sum_{ \mathrm{S_1, S_2, S_3}\text{ --- наборы путей}}(-1)^{\#\mathrm{S_1}+\#\mathrm{S_2}+\#\mathrm{S_3}}2^{-|\mathrm{S_1}|-|\mathrm{S_2}|-|\mathrm{S_3}|}.$$

%$$P(B_{3, s, 3}) =-\frac{1}{|\Omega_{3, s}|} \sum_{\mathrm S_1, \mathrm S_2, \mathrm S_3 \text{ --- наборы путей}}(-1)^{\#\mathrm S_1+\#\mathrm S_2+\#\mathrm S_3}2^{-|\mathrm S_i\cup\mathrm S_j|}4^{-|\mathrm S_i\setminus(\mathrm S_j\cup\mathrm S_k)|}.$$

Зафиксируем \(s\) и \(n = 3\). Для набора путей \(\mathrm{S} = \{s_1, \ldots, s_k\}\) обозначим множество клеток, по которым эти пути проходят, через \(\Cells(\mathrm S)\eqdef s_1\cup\ldots\cup s_k\). Для каждой тройки наборов путей \(\{\mathrm S_1, \mathrm S_2, \mathrm S_3\}\) обозначим \(T_0\hm=M_s\setminus(\Cells(\mathrm S_1) \cup \Cells(\mathrm S_2)\cup\Cells(\mathrm S_3))\), \(T_1\hm=\bigcup\limits_{\sigma\in S_3}\Cells(\mathrm S_{\sigma(1)})\setminus(\Cells(\mathrm S_{\sigma(2)})\cup\Cells(\mathrm S_{\sigma(3)}))\) и \(T_{2-3}\hm= \Cells(\mathrm S_1)\cup\Cells(\mathrm S_2)\cup\Cells(\mathrm S_3)\setminus T_1\).
\begin{idea}
	Верно следующее равенство:
%	\begin{equation*}
%	p_s = -\sum_{ \mathrm{S}\text{ --- набор путей}}(-1)^{\#\mathrm{S}}2^{-|\mathrm{S}|}.
%	\end{equation*}
	%	Более того, 
	%	\begin{equation}
	%	P(B_{n, s, n}) =-\frac{1}{2^m} \sum_{\mathrm S_1, \ldots, \mathrm S_n \text{ --- наборы путей}}(-1)^{\#\mathrm S_1+\cdots+\#\mathrm S_n}\prod_{x\in M_s}\Col(n, \sum_{i=1}^n\chi_{\mathrm S_i}(x)).
	%	\end{equation}
	%	В частности, для \(n = 3\)
	\begin{equation*}
	P(B_{3, s, 3}) = -\sum_{\mathrm S_1, \mathrm S_2, \mathrm S_3\text{ --- наборы путей}}(-1)^{\#\mathrm S_1+\#\mathrm S_2+\#\mathrm S_3}2^{-|T_1|}4^{-|T_{2-3}|}.
	\end{equation*}
\end{idea}
\begin{proof}
	Доказательство аналогично предыдущему.
	Рассмотрим все четверки \((f, \mathrm S_1, \mathrm S_2, \mathrm S_3)\) из раскраски \(f\) в четыре цвета и наборов путей \(\mathrm S_1, \mathrm S_2, \mathrm S_3\) такие, что все клетки, по которым проходят пути из набора \(S_i\), пропускают \(i\)-ую жидкость, то есть для каждой клетки \(v\hm\in \Cells(S_i)\) выполнено \(f_i(v)\hm=1\) для \(i\in\{1,2,3\}\). Заметим, что для каждой такой четверки первые три жидкости протекают от центра до границы для раскраски \(f\), так как наборы путей \(\mathrm S_i\) не пусты по определению. Припишем каждой четверке \((f, \mathrm S_1, \mathrm S_2, \mathrm S_3)\) знак \((-1)^{\#\mathrm S_1+\#\mathrm S_2+\#\mathrm S_3}\). Посчитаем сумму знаков всех четверок двумя способами.
	
	Для данной тройки наборов путей \((\mathrm S_1, \mathrm S_2, \mathrm S_3)\) есть \(2^{|T_1|}4^{|T_0|}\) четверок \((f, \mathrm{S})\), так как есть \(|T_0|\) клеток, не принадлежащих ни одному из путей, каждую из которых можно раскрасить в \(4\) цвета и \(|T_1|\) клеток, принадлежащих путям ровно из одного набора, каждую из которых можно раскрасить в \(2\) цвета. Поэтому сумма знаков всех рассматриваемых пар равна
	\begin{equation}\label{eqt::longcountingn3}
	\sum_{ \mathrm{S}\text{ --- набор путей}}(-1)^{\#\mathrm{S}}2^{|T_1|}4^{|T_0|}.
	\end{equation} 
	С другой стороны, для фиксированной раскраски \(f\), если какая-либо жидкость не протекает от центра до границы для \(f\), то нет ни одной четверки вида \((f, \mathrm S_1, \mathrm S_2, \mathrm S_3)\) среди рассматриваемых. Пусть теперь каждая жидкость протекает до границы. Пусть есть ровно \(k_1, k_2, k_3\) различных несамопересекающихся путей \(s\) таких, что для всех клеток \(v\in s\) первая, вторая и, соответственно третья жидкость протекает через эти клетки, то есть \(f_1(v)=1, f_2(v)=1\) или \(f_3(v)=1\) соответственно. Тогда четверок \((f, \mathrm S_1, \mathrm S_2, \mathrm S_3)\) таких, что \(\#\mathrm S_1 = l_1,  \#\mathrm S_2 = l_2 \text{ и } \#\mathrm S_2 = l_3\) среди рассматриваемых ровно \({k_1\choose l_1}{k_2\choose l_2}{k_3\choose l_3}\). Тогда, применив равенство \eqref{eqt:BinomEq} три раза, получим, что для фиксированного \(f\) сумма знаков всех четверок \((f, \mathrm S_1, \mathrm S_2, \mathrm S_3)\) равна
	\begin{equation*}
	\sum_{l_1=1}^{k_1}\sum_{l_2=1}^{k_2}\sum_{l_3=1}^{k_3}(-1)^{l_1+l_2+l_3}{k_1\choose l_1}{k_2\choose l_2}{k_3\choose l_3} = \sum_{l_1=1}^{k_1}\sum_{l_2=1}^{k_2}(-1)^{l_1+l_2}{k_1\choose l_1}{k_2\choose l_2}(-1) = \cdots =-1.
	\end{equation*}
	Тогда сумма знаков всех рассматриваемых пар равна \(-4^{m}P(B_{3,s,3})\). Приравняв это число выражению \eqref{eqt::longcountingn3} и разделив на \(-4^{m}\), получаем требуемое равенство.
\end{proof}

\section{Численный эксперимент}\label{sec::exper}

Нами была написана программа, перебирающая случайный набор раскрасок множества \(M_s\) без центральной клетки в \(2^{n-1}\) цветов и проверяющая протекание каждой жидкости в каждой из раскрасок. После чего подсчитывалась количество протекающих жидкостей. Код с комментариями можно найти в \cite{MyCode}.

\textit{Приблизительной вероятностью события} \(B_{n, s, k}\) считается доля раскрасок из числа рассмотренных, для которых протекло ровно \(k\) жидкостей. Так как рассматриваемые раскраски выбирались равновероятно из всего множества раскрасок, то примерные вероятности событий действительно близки к \(P(B_{n,s,k})\) по закону больших чисел.

\begin{wrapfigure}[17]{r}{260pt}\label{fig::hp_ind_n3}
	\includegraphics[scale=0.67]{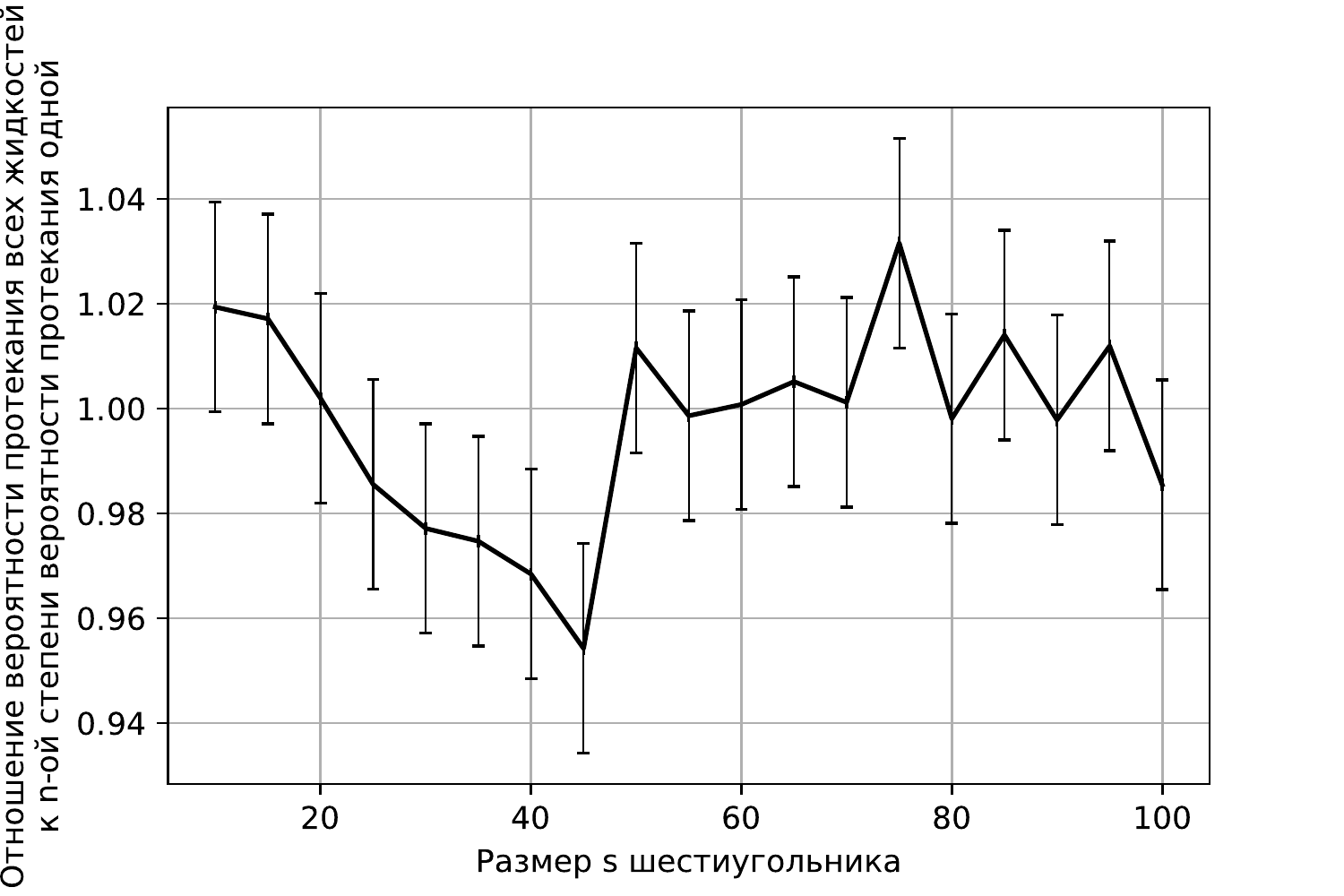}
	\caption{К задаче \ref{hyp::IndependInLimit} при \(n=3\).
		Вероятность отклонения больше указанного не превышает \(0.05\).}
\end{wrapfigure}

К сожалению, точный подсчет значений вероятностей представляется невозможным для современных компьютеров. Например, для \(s=5\) и \(n=3\) нужно покрасить \(60\) клеток в \(4\) цвета. Тогда существует \(2^{120}\hm>10^{36}\) различных раскрасок. %Предполагая, что компьютер будет анализировать \(10^{25}\) раскрасок в секунду (что является чрезвычайно оптимистичной оценкой), получаем, что на анализ всех раскрасок будет потрачено более тысячелетия.

Для начала было проверено утверждение основной теоремы \ref{thm::LikeCLT}. Были построены графики плотности вероятности для \(s=100\) и \(n=1,2,\ldots,25\). На полученном рисунке \ref{fig::MainRes} также изображена плотность вероятности нормального распределения.
Также были построены графики функции \(\dfrac{P(\bigcap_{i=1}^n A_{n, s, i})}{p_s^n}\), измеряющей степень зависимости событий протекания в совокупности, для разных \(s\). Как видно на рисунке \ref{fig::hp_ind_n3} отклонение этого отношения от \(1\) составляет не более \(0.05\) при \(n=3\). Аналогичные результаты получаются и для других значений \(n\), но они менее наглядные, так как с ростом \(n\) точность вычислений падает. Так, например, для \(n = 11\) наибольшее отклонение от \(1\) составляет уже \(0.14\).

На рисунке \ref{fig::s_to_infty} показаны графики вероятности \(P(B_{3, s, k})\) в зависимости от \(s\). Видно, что эти значения стремятся к нулю не сразу, а перед этим они ``меняют взаимное расположение''. А именно, для достаточно больших \(s\), чем больше \(k\), тем меньше вероятность, что протечет ровно \(k\) жидкостей из \(n\), что соответствует нашей интуиции.

На основе полученных результатов сформулированы следующие гипотезы и задачи.
\begin{problem}\label{hyp::IndependInLimit}
	Верно ли, что события протекания различных жидкостей становятся независимы в совокупности в пределе \(s\to\infty\), то есть для всех \(n\geqslant3\)
	\begin{equation*}
	\lim\limits_{s\to\infty}\frac{P(\bigcap_{i=1}^n A_{n, s, i})}{p_s^n} = 1?
	\end{equation*}
\end{problem}

Заметим, что любой поднабор из не более чем \(n-1\) таких событий независим в совокупности по лемме \ref{lem::UnionIndependance}.

\begin{problem}
	Верно ли, что для всех \(0\leqslant k\leqslant n\geqslant3\) имеем \(P(B_{n,s,k})\sim{n\choose k}p_s^k(1-p_s)^{n-k}\) при \(s\to\infty\)?
\end{problem}

Против утверждения задачи \ref{hyp::IndependInLimit} приводится аргумент К. Изъюрова в статье Новикова \cite[\S 4.4]{Novikov}. Мы не приводим этот аргумент, так как для понимания он требует прочтения статьи \cite{Novikov} целиком.
%\begin{wrapfigure}[20]{r}{120pt}
%	\includegraphics[scale=0.7]{liq_indep_7.pdf}
%	\caption{\(n=7\)}\label{fig::hp_ind_n7}
%	\includegraphics[scale=0.7]{liq_indep_11.pdf}
%	\caption{\(n=11\)}\label{fig::hp_ind_n11}
%\end{wrapfigure}

\begin{hypo} \label{hp::grid}
	Для всех \(n, i<j\) и достаточно больших \(s\)  имеем \(P(B_{n, s, i})>P(B_{n, s, j})\).
\end{hypo}

\section{Благодарности}
Автор выражает благодарность своему научному руководителю, Скопенкову Михаилу Борисовичу, за плодотворные обсуждения и помощь в написании статьи. Также автор благодарит Л.~Петрова, прочитавшего данную работу и приславшего ценные комментарии, и М.~Христофорова за ценные замечания.

\end{document}